\documentclass[letterpaper, 11pt, onecolumn]{ieeeconf}  
\IEEEoverridecommandlockouts                              

\usepackage{graphics} 
\usepackage{amsmath} 
\usepackage{amssymb}  
\usepackage[noadjust,sort]{cite}

\usepackage{array,arydshln}

\DeclareSymbolFont{bbold}{U}{bbold}{m}{n}
\DeclareSymbolFontAlphabet{\mathbbold}{bbold}

\title{\LARGE \bf
Output Feedback Stabilization of Switched Linear Systems with Limited Information
}
\author{Masashi Wakaiki and Yutaka Yamamoto\thanks{
This work was supported by The Kyoto University Foundation.
M. Wakaiki and Y. Yamamoto are with the Department of Applied Analysis and Complex
Dynamical Systems, Graduate School of Informatics, Kyoto University, Kyoto
606-8501, Japan
(e-mail:{\tt  \ masashiwakaiki@ece.ucsb.edu};
{\tt \ yy@i.kyoto-u.ac.jp}).}%
}

\newtheorem{theorem}{Theorem}[section]

\newtheorem{lemma}[theorem]{Lemma}

\newtheorem{assumption}[theorem]{Assumption}

\newtheorem{remark}[theorem]{Remark}

\usepackage{graphicx}

\usepackage{mathtools}
\mathtoolsset{showonlyrefs=true}
\mathtoolsset{centercolon}

\begin{document}

\maketitle
\thispagestyle{empty}
\pagestyle{empty}
\begin{abstract}                
We propose an encoding and control strategy for 
the stabilization of switched systems with limited information, 
supposing that the controller is given for each mode.
Only the quantized output and the active mode of the plant 
at each sampling time are transmitted to the controller. 
Due to switching,
the active mode of the plant may be different from that of the controller
in the closed-loop system.
Hence if switching occurs, 
the quantizer must recalculate a bounded set containing the estimation error
for quantization at the next sampling time.
We establish the global asymptotic stability
under a slow-switching assumption on dwell time and average dwell time.
To this end, we construct multiple discrete-time Lyapunov functions with
respect to the state and the size of the bounded set.
\end{abstract}

\section{Introduction}
Digital devices such as samplers, quantizers, and communication channels
play an indispensable role in low-cost, intelligent control systems.
This has motivated researchers to study control problems 
with limited information
due to sampling and quantization, 
as surveyed in~\cite{Nair2007,Hepsanha2007, Ishii2012}.
On the other hand, many systems encountered in practice have switching among
several modes of operation.
The stabilization problem of switched systems has also been studied 
extensively; see the book \cite{Liberzon2003Book},
the survey~\cite{Shorten2007, Lin2009}, and many references therein.

Both sampling/quantization and switching are
discrete-time dynamics and often appear in control systems simultaneously.
The authors of
\cite{Nair2003, Ling2010, Xiao2010, Xu2013} have studied 
quantized control for Markov jump discrete-time systems.
In \cite{Cetinkaya2012}, the stabilization of
Markov jump systems with uniformly sampled mode information is investigated.
However,  for switched systems with deterministic switching signals, 
most works deal with sampling/quantization and switching separately.
Based on the result in \cite{Liberzon2003Automatica},
our previous work~\cite{WakaikiMTNS2014} has developed 
an output encoding strategy for switched system
under an average dwell-time condition~\cite{Hespanha1999CDC} but
have not considered sampling.

The following difficulty arises from partial knowledge of the switching signal
due to sampling:
Switching can lead to the mismatch of the active modes between
the plant and the controller.
Accordingly, 
we need to prepare for another encoding strategy in case switching occurs.
For the quantization at the next sampling time, an encoding strategy
after a switch happens
must include the estimation of intersample information, e.g., the state behavior
in the sampling interval, from the transmitted data.

For switched systems with sampling and quantization, 
{\em state} feedback stabilization
has been studied under a slow-switching assumption
in \cite{Liberzon2014,Wakaiki2014IFAC}.
By contrast,
we assume that the information on 
the quantized {\em output} and the active mode of the plant
is transmitted to the controller at each sampling time.
The objective of this paper is to
develop an encoding and control strategy
achieving global asymptotic stabilization for given state feedback gains. 
The detection of switching within each sampling interval requires
a dwell-time assumption. On the other hand, we also use an average dwell-time
assumption for the convergence of the state to the origin.

Our proposed method can be seen as the extension of
\cite{Liberzon2014} from state feedback to output feedback
and also
that of \cite{Liberzon2003} from non-switched systems to switched systems.
A data-rate bound derived from our result is that from~\cite{Liberzon2003}
maximized over all the subsystems.


We organize this paper as follows. In Section II, first we show the switched 
linear system and the information structure we consider. After placing 
some basic assumptions, we state the main result.
Section III is devoted to the so-called ``zooming-out'' stage, whose objective
is to measure the output adequately.
In Section IV, we provide the encoding and control strategy 
that makes the state converges to the origin, and 
obtain a bound on the set in which the estimation error can reach when
a switch occurs.
In Section V, we show that the Lyapunov stability is achieved.
and Section VI contains a numerical example. 
Finally we conclude this paper in Section VII.

{\it Notation:~}
Let $\mathbb{Z}_+$ be the set of non-negative integers.
For $t \in \mathbb{R}$, $\lfloor t\rfloor$ is the largest integer not 
greater than $t$.

Let $\lambda_{\min}(P)$ and $\lambda_{\max}(P)$ denote 
the smallest and the largest eigenvalue of $P \in \mathbb{R}^{\sf n\times n}$.
Let $M^{\top}$ denote the transpose of $M \in \mathbb{R}^{\sf m\times n}$.

The Euclidean norm of $v \in \mathbb{R}^{\sf n}$ is
denoted by $|v| = (v^*v)^{1/2}$. 
The Euclidean induced norm of $M \in \mathbb{R}^{\sf m\times n}$ is defined by
$\|M\| = \sup \{  |Mv |:~v\in \mathbb{R}^{\sf n},~|v|= 1 \}$.
For $v = [v_1~\!\dotsb ~\!v_{\sf n} ]^{\top} \in \mathbb{R}^{\sf n}$,
its maximum norm is $|v|_{\infty} = \max\{|v_1|,\dots, |v_{\sf n}|\}$, and
the corresponding induced norm of $M \in \mathbb{R}^{\sf m\times n}$ is given by
$\|M\|_{\infty} = \sup \{  |Mv |_{\infty}:~
v\in \mathbb{R}^{\sf n},~|v|_{\infty}= 1 \}$.

\section{Output Stabilization of Switched Systems with
Limited Information}
\subsection{Switched Systems and Information Structure}
Consider the switched linear system
\begin{equation}
\label{eq:SLS}
\dot x = A_{\sigma}x+B_{\sigma}u,\quad y = C_{\sigma}x,
\end{equation}
where $x(t) \in \mathbb{R}^{\sf{n}}$ is the state,
$u(t) \in \mathbb{R}^{\sf{m}}$ is the control input, and
$y(t) \in \mathbb{R}^{\sf{p}}$ is the output.
For a finite index set $\mathcal{P}$,
the function $\sigma :~[0,\infty) \to \mathcal{P}$ is right-continuous and
piecewise constant.
We call $\sigma$ {\em switching signal} and
the discontinuities of $\sigma$ {\em switching times}.
Let $N_{\sigma}(t,s)$ stand for their number in the interval $(s,t]$. 

To generate the control input $u$, we can use
the following information on the output $y$ and the switching signal $\sigma$:

{\sl Sampling:~}
Let $\tau_s>0$ be the sampling period.
The output $y$ and the switching signal $\sigma$ are measured only at
sampling times
$k\tau_s$ $(k \in \mathbb{Z}_+)$.

{\sl Quantization:~}
Pick an odd positive number $N$.
The measured output $y(k\tau_s)$ is encoded by an integer in
$\{ 1,2\dots,N^{\sf{p}} \}$.
This encoded output and the sampled switching signal $\sigma(k\tau_s)$
are transmitted to the controller.

For the Lyapunov stability in Section V, we take $N$ to be odd.
Fig.~\ref{fig:SSSQOF} shows the closed loop we consider.
 \begin{figure}[tb]
 \centering
 \includegraphics[width = 7cm,clip]{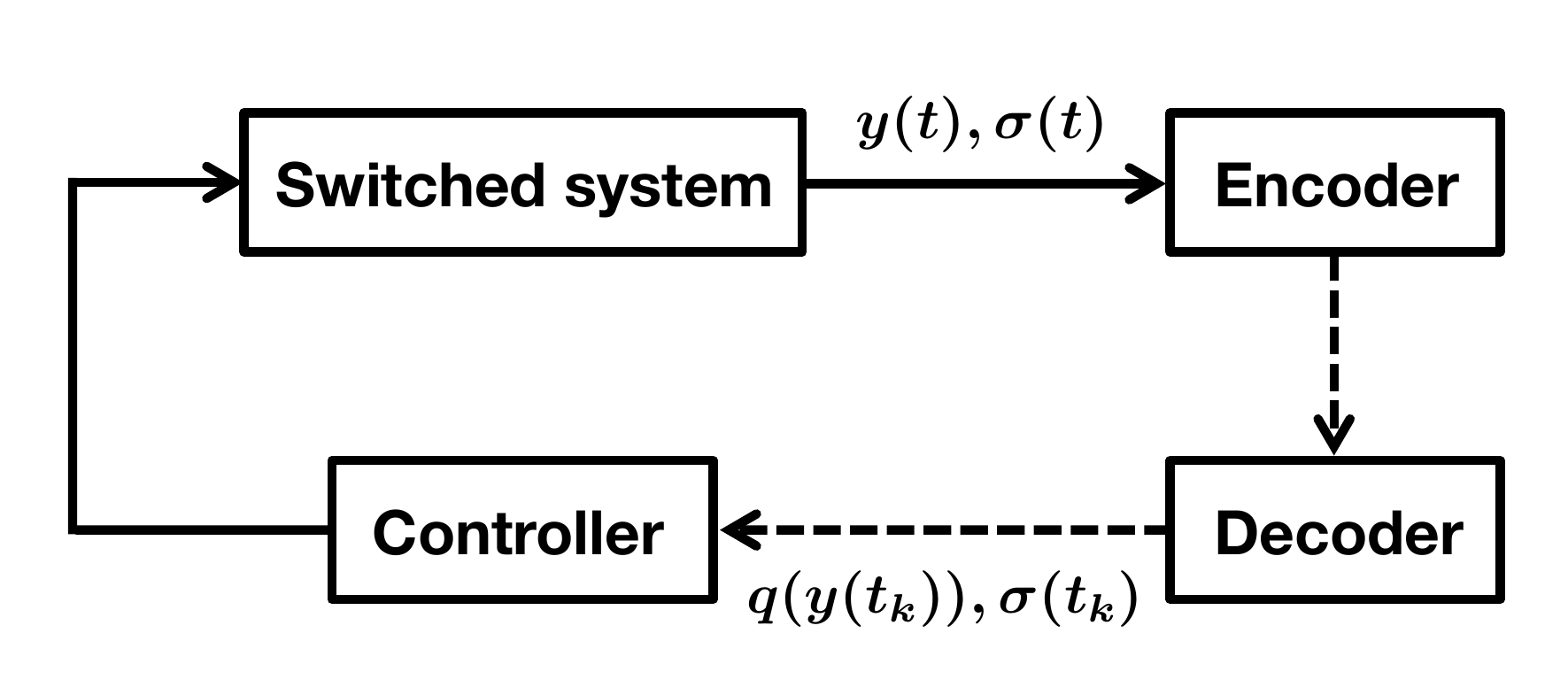}
 \caption{Sampled-data switched system with quantized output feedback}
 \label{fig:SSSQOF}
 \end{figure}

\subsection{Main Result}
Our first assumption is the stabilizability and observability of each subsystem.
\begin{assumption}
\label{ass:system}
{\em
For every $p \in \mathcal{P}$, $(A_p, B_p)$ is stabiliable and
$(C_p, A_p)$
is observable.
We choose $K_p \in \mathbb{R}^{\sf m \times n}$ so that
$A_p+B_pK_p$ is Hurwitz.
For all $A_p$. 
the sampling time $\tau_s$ is not pathological.
}
\end{assumption}
The non-pathological sampling time implies that
$(C_p, e^{A_p\tau_s})$ is observable in the discrete-time sense,
which is used for state reconstruction.

Next we assume that the switching signal $\sigma$ has
the following slow-switching properties:
\begin{assumption}
\label{ass:switching_time}
{\em
{\sl Dwell time:~}
Every interval between two switches is not smaller than
the sampling period $\tau_s$. That is, 
$N_{\sigma}(t,s) \leq 1$ if $t-s \leq \tau_s$.

{\sl Average dwell time~\cite{Hespanha1999CDC}:~}
There exist $\tau_a>0$ and $N_0 \in \mathbb{N}$ such that
\begin{equation}
\label{eq:ADT_cond}
N_{\sigma}(t,s) \leq N_0 + \frac{t-s}{\tau_a}
\end{equation}
for all $t > s \geq 0$.
}
\end{assumption}

Switching signals in
Assumption \ref{ass:switching_time} are called 
{\em hybrid dwell-time} signals~\cite{Vu2011, Liberzon2014}.
The assumption on dwell time is necessary for the
detection of a switch between sampling times,
while that on average dwell time is used in the proof that the state
converges to the origin.

Furthermore, we extend the quantization assumption for 
systems with a single mode in \cite{Liberzon2003} to
switched systems.
\begin{assumption}
\label{ass:quantization}
{\em
Let $\eta_p$ be the smallest natural number such that
$W_p$ defined by
\begin{equation}
\label{eq:WpDef}
W_p = 
\begin{bmatrix}
C_p \\ C_p e^{A_p\tau_s} \\ \vdots \\
C_p e^{A_p(\eta_p - 1)\tau_s}
\end{bmatrix}
\end{equation}
has full column rank.
Let $W_p^{\dagger}$ be a left inverse of $W_p$.
Then 
\begin{equation}
\label{eq:N_condition}
\|e^{A_p \eta_p \tau_s}W_p^{\dagger}\|_{\infty}
\cdot \max_{0\leq k \leq \eta_p-1}\left\|C_pe^{A_p k \tau_s}\right\|_{\infty} < N
\end{equation}
for all $p \in \mathcal{P}$.
}
\end{assumption}

Assumption \ref{ass:quantization} gives a lower bound on the available data rate
implicitly, and
\eqref{eq:N_condition} is the data-rate bound from \cite{Liberzon2003} 
maximized over the individual modes.
This assumption is used for finer quantization when a switch does not occur.
Note that as $\tau_s$ becomes small,
$e^{A_p \eta_p \tau_s}$ and $C_pe^{A_p k \tau_s}$
converge to $I$ and $C_p$ respectively, but that $W_p$ does not have full column 
rank in general
if $\tau_s = 0$.
Therefore
the left side of \eqref{eq:N_condition} may not decrease 
as $\tau_s$ tends to zero.

If $C_p = I$, then $W_p = I$ and $\eta_p = 1$.
Hence \eqref{eq:N_condition} is consistent to
the data rate assumption in the state feedback case \cite{Liberzon2014}.

The main result shows that
global asymptotic stabilization is possible if
the average dwell time is sufficiently large.
\begin{theorem}
\label{thm:stability_theorem2}
{\em
Consider the switched system \eqref{eq:SLS}, and let 
Assumptions \ref{ass:system}, \ref{ass:switching_time}, and 
\ref{ass:quantization} hold.
If the average dwell time $\tau_a$ in \eqref{eq:ADT_cond} 
is larger than a certain value, then
there exists an output encoding that achieves the following stability
for every $x(0) \in \mathbb{R}^{\sf n}$ and every $\sigma(0) \in \mathcal{P}$:

\noindent
{\sl Convergence to the origin:~}
\begin{equation}
\label{eq:convergence_x}
\lim_{t \to \infty}x(t) = 0.
\end{equation}

\noindent
{\sl Lyapunov stability:~}
To every $\varepsilon > 0$, there corresponds $\delta > 0$ such that
if $|x(0)| < \delta$, then $|x(t)| < \varepsilon$
for $t \geq 0$.
}
\end{theorem}

A constructive proof of Theorem \ref{thm:stability_theorem2} is given 
in the next sections.
We obtain a sufficient condition \eqref{eq:ADT_final_condition} on $\tau_a$
in the proof.
As in \cite{Liberzon2003, WakaikiMTNS2014, Liberzon2014},
we show the convergence to the origin 
by dividing the proof into the ``zooming-out'' stage and
the ``zooming-in'' stage.

\section{Zooming-out Stage}

The objective of the ``zooming-out'' stage is to generate 
an upper bound on the estimation error of the state.
We have to obtain such a bound by using the quantized output
and the switching signal at each sampling time.

Define $\eta$ by
\begin{align}
\label{eq:eta_def}
\eta = \max_{p \in \mathcal{P}} \eta_p.
\end{align}

At this stage, we set the control input $u = 0$.
Assume that the average dwell time $\tau_a$ satisfies
\begin{equation}
\label{eq:DTC_zo}
\tau_a > \eta \tau_s.
\end{equation}

Pick $\mu_0 > 0$ and $\chi > 0$, and define
\begin{equation}
\label{eq:j_def_zo}
\mu_{n} = e^{(1+\chi)\max_{p \in \mathcal{P}} \|A_p\|_{\infty} n\tau_s} \mu_0
\end{equation}
for $n\in \mathbb{Z}_+$. We construct the encoding function $Q_{n}$ by 
\begin{equation*}
Q_n(y) =
\begin{cases}
0
&  \quad \text{if~~}y(n \tau_s) \in
\{ y \in \mathbb{R}^{\sf p}:~| y |_{\infty} \leq \mu_n\} \\ 
1
&  \quad \text{otherwise}.
\end{cases}
\end{equation*}
The following theorem is used 
for the reconstruction of the state:
\begin{theorem}
\label{prop:property_of_ADT}
{\em
If the average dwell time $\tau_a$ in \eqref{eq:ADT_cond} 
satisfies \eqref{eq:DTC_zo}, then
there exists an integer $n_0 \geq 0$ such that
\begin{gather}
Q_{n}(y) = 0 \label{eq:zooming_out_QY_SD} \\
\sigma(n\tau_s) = \sigma(n_0\tau_s) =: p \label{eq:zooming_out_SS_SD}
\end{gather}
for $n =n_0,n_0+1,\dots, n_0 + \eta_p - 1$. 
Such $n_0$ satisfies $n_0 \leq n_1$, where
$n_1$ depends on $N_0$ and $\tau_a$ in \eqref{eq:ADT_cond}
but not on $\sigma$ itself.
}
\end{theorem}

To prove Theorem \ref{prop:property_of_ADT}, 
we use the following property of average dwell time:
\begin{lemma}[\cite{WakaikiMTNS2014}]
\label{lem:ADT_upperbound}
{\em
Fix an initial time $t_0 \geq 0$.
Suppose that $\sigma$ satisfies \eqref{eq:ADT_cond}.
Let $\tau_0 \in (0, \tau_a)$, and choose an integer $m$ such that
\begin{equation}
\label{eq:N_ADTcond}
m > \frac{\tau_a}{\tau_a - \tau_0} \left( N_0 - \frac{\tau_0}{\tau_a} \right).
\end{equation}
There exists $T \in [t_0, t_0+(m-1)\tau_0]$ such that 
$N_{\sigma}(T+\tau_0,T) = 0$.
}
\end{lemma}

\begin{proof}[Proof of Theorem 3.1.]
The growth rate of $\mu_n$ in \eqref{eq:j_def_zo} is 
larger than that of $|y|_{\infty}$ for arbitrary switching. Hence
there is an integer $\bar{n}_0$ such that $| y |_{\infty} \leq \mu_{n}$
for all $n \geq \bar{n}_0$, which leads to \eqref{eq:zooming_out_QY_SD}.

Let $m$ be an integer satisfying \eqref{eq:N_ADTcond} with 
$\eta \tau_s$ in place of $\tau_0$.
Since $\tau_a > \eta \tau_s$,
Lemma \ref{lem:ADT_upperbound} shows that
$
N_{\sigma}(T + \eta \tau_s,T) = 0
$
for some $T \in [\bar{n}_0 \tau_s, (\bar{n}_0 + (m-1)\eta ) \tau_s ]$.
The interval $(T, T+ \eta \tau_s]$ contains $\eta$ sampling times. Thus
we have an integer $n_0 \in [\bar{n}_0, \bar{n}_0 + (m-1) \eta]$ satisfying
\eqref{eq:zooming_out_QY_SD} and \eqref{eq:zooming_out_SS_SD}
for $n =n_0,n_0+1,\dots, n_0 + \eta_p - 1$, and $n_0 \leq n_1 := 
\bar{n}_0 + (m-1) \eta$.
\end{proof}


In conjunction with the dwell-time assumption, 
\eqref{eq:zooming_out_SS_SD} shows that
the active mode of the plant does not change in 
$[n_0\tau_s, (n_0 + \eta_p -1)\tau_s]$.
We can therefore reconstruct $x(n_0\tau_s)$ by using
$W_p$ in \eqref{eq:WpDef} and the output at 
$t = n_0\tau_s,\dots , (n_0 + \eta_p -1)\tau_s$:
\begin{equation}
\label{eq:x_j0}
x(n_0\tau_s) = W_p^{\dagger}
\begin{bmatrix}
y(n_0\tau_s) \\
\vdots \\
y((n_0 + \eta_p - 1)\tau_s)
\end{bmatrix}.
\end{equation}
The rest of the procedure is the same as in the non-switched 
case \cite{Liberzon2003}.
Combining
\eqref{eq:zooming_out_QY_SD} and \eqref{eq:x_j0}, 
we obtain
\begin{equation}
\label{eq:En0_def}
|x(n_0 \tau_s)|_{\infty} \leq  \|W_p^{\dagger}\|_{\infty}\cdot
\mu_{n_0+\eta_p - 1} =: E_{n_0}.
\end{equation}
It follows that
\begin{align}
|x((n_0 + \eta_p)\tau_s)|_{\infty} 
\leq 
e^{\max_{p \in \mathcal{P}} \|A_p\|_{\infty} \tau_s} 
&\left\|
e^{A_p(\eta_p - 1)\tau_s}
\right\|_{\infty} 
E_{n_0} 
=: E_{n_0+\eta_p}.
\label{eq:En0+etap}
\end{align}
Define the estimated state $\xi$ at $t =(n_0+\eta_p)\tau_s$ by
\begin{equation}
\label{eq:ES_Initial}
\xi((n_0+\eta_p)\tau_s) = 0.
\end{equation}
Then the error $e=x -\xi$ satisfies
$
|e((n_0+\eta_p)\tau_s)|_{\infty} \leq E_{n_0+\eta }.
$
This completes the ``zooming-out'' stage.

\section{Zooming-in Stage}
Here we construct an encoding and control strategy for the
convergence of the state to the origin.
Since the size $E_k$ of the quantization region increases
after a switch occurs, the term ``zooming-in'' may be misleading.
However, in order to contrast the ``zooming-out'' phase in the
previous section, 
we call the stage in this section the ``zooming-in'' stage as in
\cite{Liberzon2003Automatica,Brockett2000}.

Let $t_0=k_0\tau_s \geq 0$ be the initial time of the zooming-in stage
or the time at which the upper bound $E_k$ of the estimation error is updated. 
Let $\xi$ and $e$ be the estimated state and the estimation error $x - \xi$,
respectively.
Assume that $\sigma(k_0\tau_s) = p$ and 
$|e(k_0\tau_s)|_{\infty} \leq E_{k_0}$.

\subsection{Basic encoding and control method}
If no switch happens, then
we can use the encoding and control method
for systems with a single mode in \cite{Liberzon2003}.
However, after a switch occurs, a modified upper bound on 
the estimation error is needed for the next quantized measurement.
We shall obtain the upper bound in 
Section IV. C. 1).
In this subsection, assuming that the state estimate $\xi(k_0)$ and
the error bound $E_{k_0}$
are derived, we
briefly state the encoding and control method 
because it will be needed in the sequel.

Let $\sigma(k_0\tau_s) = p$.
If no switch occurs in $(t_{k_0}, t_{k_0+\eta_p \tau_s}]$,
we set $k=\eta_p$, and otherwise we define $k$
by the minimum integer in the interval $[1,\eta_p]$ such that 
$\sigma((k_0+k-1)\tau_s) \not= \sigma((k_0+k)\tau_s)$.
We generate the state estimate $\xi$ and the output estimate $\hat y$ by
\begin{equation}
\label{eq:xi_deq}
\dot \xi = (A_p + B_p K_p)\xi,\quad
\hat y = C_p \xi
\end{equation}
for $t \in [k_0\tau_s, (k_0+k) \tau_s)$, and set
the control input 
\begin{equation}
\label{eq:u_def}
u = K_p\xi.
\end{equation}
Since
$
\dot x 
= A_p x + B_p K_p \xi,
$
it follows that $e = x - \xi$ satisfies
\begin{equation}
\label{eq:e_deq}
\dot e = A_p e.
\end{equation}
If $l=0,\dots,k-1$, then
\begin{align*}
| y((k_0+l)\tau_s) - \hat y((k_0+l)\tau_s) |_{\infty}
&\leq \left\|C_p e^{A_p l \tau_s} \right\|_{\infty}
E_{k_0}.
\end{align*}
For $l=0,\dots,k-1$,
divide the hypercube 
\begin{equation}
\label{eq:quantization}
\left\{ y \in \mathbb{R}^{\sf p}:~| y  - \hat y((k_0+l)\tau_s)|_{\infty} 
\leq \left\|C_p e^{A_p l \tau_s} \right\|_{\infty}
E_{k_0} \right\}
\end{equation}
into $N^{\sf p}$ equal boxes and assign a number in $\{1,\dots, N^{\sf p}\}$
to each divided box by a certain one-to-one mapping.
The encoder sends to the decoder 
the number $Q_{k_0+l}$ of the divided box containing 
$y((k_0+l)\tau_s)$, and then
the decoder generates $q_{k_0+l}$ equal to the center of the box
with number $Q_{k_0+l}$.
If $y((k_0+l)\tau_s)$ lies on the boundary on several boxes, then
we can choose any one of them.

\subsection{Non-switched case}
The calculation of an upper bound $E_k$ on the estimation error
is dependent of whether 
a switch occurs in an interval with length $\eta_p \tau_s$.
Let us first study the case without switching in the interval
$(k_0\tau_s, (k_0+\eta_p) \tau_s]$, i.e., the case 
$\sigma(k_0\tau_s)  = \dots = \sigma((k_0+\eta_p) \tau_s) =: p$.

\subsubsection{Calculation of an error bound}
An upper bound $E_{k_0 + \eta_p}$ on $|e((k_0+\eta_p)\tau_s)|_{\infty}$ can be
obtained in the same way as in \cite{Liberzon2003}.
We therefore omit the details of the calculation here.

Define
\begin{equation}
\label{eq:gamma_p_def}
\theta_p = \frac{\|e^{A_p \eta_p \tau_s}W_p^{\dagger}\|_{\infty}
\cdot \max_{0\leq k \leq \eta_p-1}\left\|C_pe^{A_p k \tau_s}\right\|_{\infty}}{N}.
\end{equation}
From the result in \cite{Liberzon2003}, 
if we appropriately determine $\xi$ at $t=(k_0+\eta_p)\tau_s$
from the transmitted data $q_{k_0},\dots,q_{k_0+\eta_p-1}$,
then we obtain
$|e(t_0 + \eta_p \tau_s)|_{\infty} \leq E_{k_0+\eta_p}$, where
$E_{k_0+\eta_p}$ is defined by
\begin{equation}
\label{eq:e_noswitch_bound}
E_{k_0+\eta_p} = \theta_p E_{k_0}.
\end{equation}
Note that $\theta_p < 1$ for every $p \in \mathcal{P}$ 
by \eqref{eq:N_condition}.

\subsubsection{Decrease rate of multiple Lyapunov functions}
Here we construct a discrete-time Lyapunov function $V_p$ of mode $p$ 
with respect to $x(k\tau_s)$ and $E_k$.
The calculation below 
is similar to that in the state feedback case~\cite{Liberzon2014}, but
we sketch it for completeness.

For simplicity of notation, we write $V_p(k)$ instead of
$V_p(x(k\tau_s), E_k)$.
We obtain an upper bound of $V_p(k_0+\eta_p)$ using $V_p(k_0)$.

First we obtain $x((k_0+\eta_p)\tau_s)$ from $x(k_0\tau_s)$ and
$e(k_0\tau_s)$.
Since
\begin{align}
\label{eq:x_de_noswitch}
\dot x 
= A_p x + B_p K_p\xi 
= (A_p + B_pK_p) x - B_pK_p e
\end{align}
for $t \in  [k_0\tau_s, (k_0+\eta_p)\tau_s)$,
it follows from \eqref{eq:e_deq} that
\begin{align}
x((k_0+\eta_p)\tau_s) 
=
\bar A_p x(k_0\tau_s) + \bar B_p e(k_0\tau_s),
\label{x_t0+eta_tau}
\end{align}
where $\bar A_p$ and $\bar B_p$ are defined by
\begin{align*}
\bar A_p &= 
e^{(A_p+B_pK_p) \eta_p \tau_s} \\
\bar B_p &= 
\int^{\eta_p\tau_s}_{0}
e^{(A_p+B_pK_p)(\eta_p \tau_s - t)}
B_pK_p e^{A_p t} dt.
\end{align*}
Recall that $\bar A_p$ is Hurwitz by Assumption \ref{ass:quantization}.
To every positive definite matrix $Q_p$, 
there correponds a positive definite matrice $P_p$ such that
\begin{equation}
\label{eq:Lyapnov_equation}
\bar A_p ^{\top} P_p \bar A_p - P_p = -Q_p.
\end{equation}
Fix $\rho_p > 0$
for each $p \in \mathcal{P}$. Similarly to \cite{Liberzon2014},
define the Lyapunov function $V_p$ by
\begin{equation}
\label{eq:Lyapunov_def}
V_p(k) = 
V_p(x(k\tau_s),E_k)
= x(k\tau_s)^{\top}P_p x(k \tau_s) + \rho_p E_{k}^2.
\end{equation}

Pick $\kappa_p > 1$.
A simple computation gives
\begin{align}
&x((k_0+\eta_p)\tau_s)^{\top}P_px((k_0+\eta_p)\tau_s) - 
x(k_0\tau_s)^{\top}P_px(k_0\tau_s) \notag \\
&\quad \leq
- \lambda_{\min}(Q_p) |x(k_0\tau_s) |^2 
+2 \|\bar A_p^{\top} P_p \bar B_p\|\cdot  |x(k_0\tau_s) | \cdot |e(k_0\tau_s) |
+\| \bar B_p^{\top}P_p \bar B_p \| \cdot | e(k_0\tau_s)|^2 \notag \\
&\quad \leq
-\frac{1}{\kappa_p} \lambda_{\min}(Q_p) |x(k_0\tau_s) |^2 
- \frac{\kappa_p -1}{\kappa_p } \left(
\sqrt{\lambda_{\min}(Q_p)} \cdot |x(k_0\tau_s)|
-\frac{\kappa_p}{\kappa_p -1}
\frac{\|\bar A_p^{\top} P_p \bar B_p\|}{\sqrt{\lambda_{\min}(Q_p)}} | e(k_0\tau_s)|
\right)^2 \notag \\[-5pt]
&\qquad \qquad
+ 
\left( \frac{\kappa_p}{\kappa_p -1}
\frac{\|\bar A_p^{\top} P_p \bar B_p\|^2}{\lambda_{\min}(Q_p)} +
\| \bar B_p^{\top}P_p \bar B_p \|
\right)  | e(k_0\tau_s)|^2 \notag \\
&\quad \leq
-\alpha_p |x(k_0\tau_s) |^2 
 +\beta_p E_{k_0}^2,
\label{eq:x_Lyapnov}
\end{align}
where $\alpha_p$ and $\beta_p$ are defined by
\begin{align*}
\alpha_p &= \frac{1}{\kappa_p} \lambda_{\min}(Q_p) \\
\beta_p &= {\sf n}
\left( \frac{\kappa_p}{\kappa_p -1}
\frac{\|\bar A_p^{\top} P_p \bar B_p\|^2}{\lambda_{\min}(Q_p)} +
\| \bar B_p^{\top}P_p \bar B_p \|
\right).
\end{align*}
Combining \eqref{eq:e_noswitch_bound} with \eqref{eq:x_Lyapnov},
as in \cite[Lemma 1]{Liberzon2014}, we obtain
\begin{align*}
V_p(k_0+\eta_p)
&\leq
\left(
1 - \frac{\alpha_p}{\lambda_{\max}(P_p)}
\right) x(k_0\tau_s)^{\top}P_p x(k_0\tau_s) 
+ \left(\frac{\beta_p}{\rho_p} + \theta_p^2 \right) \rho_p E_{k_0}^2.
\end{align*}
Since $\theta_p<1$,
we can choose $\rho_p$ so that
\begin{equation}
\label{eq:rho_cond}
\rho_p > \frac{\beta_p}{1 - \theta_p^2}.
\end{equation}
Then defining
\begin{equation}
\label{eq:nu_p_def}
\nu_p := 
\max
\left\{
1 - \frac{\alpha_p}{\lambda_{\max}(P_p)},~~
\frac{\beta_p}{\rho_p} + \theta_p^2
\right\} < 1,
\end{equation}
we finally obtain
\begin{equation}
\label{eq:Lyapunov_nonswitched}
V_p(k_0+\eta_p) \leq \nu_p V_p(k_0).
\end{equation}
Note that $\nu_p$ depends on $\kappa_p > 1$ and $\rho_p$ with \eqref{eq:rho_cond}.
We can use these parameters to make
the encoding and control strategy less conservative, i.e., to allow 
smaller average dwell time.

\subsection{Switched case}
Next we study the case in which a switch
occurs in the interval $(k_0\tau_s, (k_0+\eta_p) \tau_s]$.
Suppose that $k \in \mathbb{N}$, with $k \leq \eta_{p}$, for which
the first switching time $T$ after $k_0\tau_s$ satisfies
$
T \in ((k_0+k-1)\tau_s, (k_0+k) \tau_s].
$
That is, 
$\sigma(k_0\tau_s) =  \dots = \sigma((k_0+k-1) \tau_s) =: p$ and 
$\sigma((k_0+k-1) \tau_s) \not= \sigma((k_0+k) \tau_s) =: q$.
In this case, the estimated state $\xi$ and 
the controller input $u$
in the interval
$[k_0\tau_s, (k_0+k) \tau_s]$ are given 
by \eqref{eq:xi_deq} and \eqref{eq:u_def},
respectively.
Note that 
the switching information is not transmitted instantly. However,
the controller can detect the switch at the next sampling time.
This is because the dwell-time condition in 
Assumption \ref{ass:switching_time} implies that at most one
switch occurs between sampling time. 

\subsubsection{Calculation of an error bound}
Our first objective here is to obtain 
an upper bound $E_{k_0+k}$ of $|e((k_0+k)\tau_s)|_{\infty}$
from the information $\xi(k_0\tau_s)$ and $E_{k_0}$
available to the quantizer.

\begin{lemma}
\label{lem:e_bound_switchedcase1}
{\em
Define $H_{p,q}$, $\bar \delta_{p,q}(k)$, and
$\bar \gamma_{p,q}'(k)$ by
\begin{align}
&H_{p,q} = (A_q - A_p) + (B_q - B_p)K_p \label{eq:Hpq} \\
&\bar \delta_{p,q}(k) = \max_{0\leq \tau \leq \tau_s}
\left\|
\int^{\tau_s}_{\tau}
e^{A_q(\tau_s - t)}H_{p,q} e^{(A_p+B_pK_p)(t+(k-1)\tau_s )} dt 
\right\|_{\infty} \notag\\
&\bar \gamma_{p,q}'(k) = \max_{0\leq \tau \leq \tau_s}
\left\| e^{A_q(\tau_s - \tau)} e^{A_p((k-1)\tau_s + \tau)}
 \right\|_{\infty}. \notag
\end{align}
Then we obtain the following upper bound of $|e(t_0+k \tau_s)|_{\infty}$:
\begin{align}
|e((k_0+k) \tau_s)|_{\infty}
\leq
\bar \delta_{p,q}(k) |\xi (k_0\tau_s)|_{\infty} 
+ \bar \gamma_{p,q}'(k) E_{k_0} =: E_{k_0+k}. \label{eq:Ek0_k_def}
\end{align}
}
\end{lemma}

\begin{proof}
Since 
$e$ is determined by
\eqref{eq:e_deq} before the switching time $T$, 
it follows that
\[
e(T) = e^{A_p(T-k_0\tau_s)} e(k_0\tau_s).
\]

Let us consider the error behavior for $t > T$.
The mode of the plant changes from $p$ to $q$ after $T$, 
while that of the controller is still $p$. We therefore have
\begin{equation}
\label{eq:x_de_switch}
\dot x = A_qx + B_qK_p \xi,
\end{equation}
and it follows that $e$ satisfies
\begin{align}
\dot e 
= A_q e + H_{p,q} \xi \label{eq:e_eq_after_switch}
\end{align}
for $t > T$, 
where $H_{p,q}$ is defined by \eqref{eq:Hpq}.


As regards $\xi$, \eqref{eq:xi_deq} gives
\begin{equation}
\label{eq:xi_with_switch}
\xi (t) = e^{(A_p+B_pK_p)(t-k_0\tau_s)} \xi(k_0\tau_s)
\end{equation}
for $t \in [k_0\tau_s, (k_0+k)\tau_s]$.
Substituting this into \eqref{eq:e_eq_after_switch}, we obtain
\begin{align*}
e((k_0+k) \tau_s) 
=
e^{A_q(\tau_s - \tau)} e^{A_p((k-1)\tau_s + \tau)} e(k_0\tau_s)+
\int^{\tau_s}_{\tau}
e^{A_q(\tau_s - t)}H_{p,q} e^{(A_p+B_pK_p)(t+(k-1)\tau_s )} dt 
\cdot \xi(k_0\tau_s), 
\end{align*}
where $\tau = T - (k_0 + k-1)\tau_s$ and $0 < \tau \leq \tau_s$. 
Thus we obtain \eqref{eq:Ek0_k_def}.
\end{proof}

\begin{remark}

{\bf (1)}
The propose method discards the quantized measurements
$q_{k_0},\dots,q_{k-1}$. If we use these data, then
a better $E_k$ can be obtained, in particular, 
in the case when a switch occurs in the last sampling interval
$((k_0+\eta_p-1)\tau_s, (k_0+\eta_p) \tau_s]$.

Here we briefly explain how to obtain the error bound
in the switched case
by using the quantized measurements
$q_{k_0},\dots,q_{k-1}$. 
For simplicity, let us assume that
the switching time $T$ is in the last sampling interval, i.e., 
$T \in ((k_0+\eta_p-1)\tau_s, (k_0+\eta_p) \tau_s]$, and let
and $\sigma((k_0+\eta_p-1)\tau_s) = p$ and 
$\sigma((k_0+\eta_p)\tau_s) = q \not= p$.
In this case, we can construct the state estimate $\zeta_0$ for $x(k_0\tau_s)$
by using the measurements $q_{k_0},\dots, q_{k_0+\eta_p-1}$.
We assume that $|\zeta_0 - x(k_0\tau_s)|_{\infty} \leq E_{\zeta}$.

Similarly to $\xi$ in \eqref{eq:xi_deq}, we define the dynamics of $\zeta$ by
\begin{equation*}
\dot \zeta = A_p \zeta + B_pu,\qquad \zeta(k_0\tau_s) = \zeta_0.
\end{equation*} 
Define $e_{\zeta} = x - \zeta$.
Recalling that $u = K_p \xi$, we can write 
the dynamics of $e_{\zeta}$ after a switch as
\begin{equation*}
\dot e_{\zeta} = 
A_q e_{\zeta} + (A_q - A_p) \zeta + (B_q - B_p) K_p \xi,
\end{equation*}
and hence
\begin{equation*}
e_{\zeta}((k_0+\eta_p)\tau_s) = 
F_{1}(T) e_{\zeta}(k_0\tau_s) + 
F_{2}(T) \zeta_0 +
F_{3}(T) \xi(k_0\tau_s)
\end{equation*}
for some continuous functions $F_1$, $F_2$, and $F_3$.
Therefore if we define the new state estimate $\xi((k_0+\eta_p)\tau_s)$ by
\begin{equation*}
\xi((k_0+\eta_p)\tau_s) := \zeta((k_0+\eta_p)\tau_s),
\end{equation*}
then the error bound $E_{k_0+\eta_p}$ is
given by
\begin{equation*}
E_{k_0+\eta_p} = 
\max_{T \in I_{0}}
\|F_1(T)\|_{\infty} \cdot E_{\zeta} + 
\max_{T \in I_{0}}
\|F_2(T)\|_{\infty} \cdot |\zeta_0|_{\infty} +
\max_{T \in I_{0}}
\|F_3(T)\|_{\infty} \cdot |\xi(k_0\tau_s)|_{\infty},
\end{equation*}
where $I_{0} := [(k_0+\eta_p-1)\tau_s, (k_0+\eta_p) \tau_s]$.


{\bf (2)}
For simplicity, we use \eqref{eq:xi_with_switch} 
for the estimated state at $t = (k_0+k)\tau_s$.
However, this estimate makes the corresponding bound $E_{k_0+k}$
be larger if the switch occurs just after the sampling time.
To avoid this conservatism, 
two auxiliary time variables can be used as in \cite[Section 4.2]{Liberzon2014}.

\end{remark}
\subsubsection{Increase rate of multiple Lyapunov functions}
Let us next find an upper bound of $V_q(k_0+k)$ described by $V_p(k_0)$.
To this end, we need upper bounds on $|x ((k_0+k)\tau_s)|$
and $E_{k_0+k}$ by using $|x (k_0\tau_s)|$ and $E_{k_0}$.

\begin{lemma}
\label{lem:Ek0+k_bound}
{\em
Define $\bar \delta_{p,q}$ and $\bar \gamma_{p,q}'$ as in 
Lemma \ref{lem:e_bound_switchedcase1} and also
$\bar \gamma_{p,q}$ by 
$\bar \gamma_{p,q}(k) = \bar \delta_{p,q}(k) + \bar \gamma_{p,q}'(k)$.
Then we have
\begin{align}
\label{eq:e_bound_mismatch}
E_{k_0+k} 
\leq \bar \delta_{p,q}(k)|x (k_0\tau_s)| + 
\bar \gamma_{p,q}(k) E_{k_0}.
\end{align}
}
\end{lemma}
\begin{proof}
This follows from the definition \eqref{eq:Ek0_k_def} of $E_{k_0+k}$ and
\[
|\xi(k_0\tau_s)|_{\infty} \leq  |x(k_0\tau_s)|_{\infty} + |e(k_0\tau_s)|_{\infty}
\leq |x(k_0\tau_s)| + E_{k_0}.
\]
\end{proof}

\begin{remark}
Note that $E_{k_0+k}$ must be determined from the available data 
$\xi(k_0\tau_s)$ and $E_{k_0}$. In contrast, 
since the variables of the Lyapunov function $V_q$ are 
$x(k_0\tau_s)$ and $E_{k_0}$,
we need an upper bound on $E_{k_0+k}$ described
by $x(k_0\tau_s)$ and $E_{k_0}$.
If we use $\xi$ as a variable of the Lyapunov functions
as in \cite{Liberzon2014},
then the conservatism in Lemma \ref{lem:Ek0+k_bound}
can be avoided. Instead of that, however, 
\eqref{eq:x_Lyapnov} becomes conservative.
\end{remark}

Now we obtain an upper bound of $|x((k_0+k)\tau_s)|$.
\begin{lemma}
{\em
Define $\bar \alpha_{p,q}(k)$ and $\bar \beta_{p,q}(k)$ by
\begin{align*}
\bar \alpha_{p,q}(k) &=
\max_{0 \leq \tau \leq \tau_s}
\left\|
e^{A_q(\tau_s - \tau)}e^{(A_p+B_pK_p)\tau(k)}
+ 
\int^{\tau_s }_{\tau}
e^{A_q(\tau_s  - t)}
B_qK_p e^{(A_p+B_pK_p) (t + (k-1)\tau_s)} dt
\right\|.  \\
\bar \beta_{p,q} (k)
&=
\sqrt{\sf n} 
\max_{0 \leq \tau \leq \tau_s}
\biggr\|
e^{A_q(\tau_s - \tau)}
\int^{\tau(k)}_{0}
e^{(A_p+B_pK_p)(\tau(k) - t)}
B_pK_p e^{A_p t} dt \notag \\
&\qquad \qquad \qquad \qquad \qquad -
\int^{\tau_s }_{\tau}
e^{A_q(\tau_s  - t)}
B_qK_p e^{(A_p+B_pK_p) (t + (k-1)\tau_s)} dt
\biggr\|.
\end{align*}
Then we derive
\begin{align}
\label{eq:x_bound_mismatch}
|x((k_0+k)\tau_s)| 
\leq \bar \alpha_{p,q}(k) |x(k_0\tau_s)| +
\bar \beta_{p,q}(k) E_{k_0}.
\end{align}
}
\end{lemma}

\begin{proof}
Recall that $x$ and $e$ satisfy \eqref{eq:x_de_noswitch} 
and \eqref{eq:e_deq}
before the switching time $T$.
Defining 
$
\tau(k) = T - k_0\tau_s,
$
we have
\begin{align}
\label{eq:xT}
x(T)
=
e^{(A_p+B_pK_p)\tau(k)} x(k_0\tau_s) +\int^{\tau(k)}_{0}
e^{(A_p+B_pK_p)(\tau(k) - t)}
B_pK_p e^{A_p t} dt \cdot
e(k_0\tau_s). 
\end{align}
On the other hand, 
since 
$x$ satisfies \eqref{eq:x_de_switch}
after the switching time $T$, it follows from \eqref{eq:xi_with_switch} that
\begin{align}
x((&k_0+k)\tau_s) 
=
e^{A_q(\tau_s - \tau)} x(T) +
\int^{\tau_s }_{\tau}
e^{A_q(\tau_s  - t)}
B_qK_p e^{(A_p+B_pK_p) (t +(k-1)\tau_s)} dt  \cdot
(x(k_0\tau_s) - e(k_0\tau_s)), \label{eq:xk0_k_tau_s}
\end{align}
where $\tau = T - (k_0 + k-1)\tau_s$ and $0 < \tau \leq \tau_s$.
Substituting \eqref{eq:xT} into \eqref{eq:xk0_k_tau_s},
we derive the desired result \eqref{eq:x_bound_mismatch}.
\end{proof}

Similarly to \cite[Lemma 2]{Liberzon2014}, 
\eqref{eq:e_bound_mismatch} and \eqref{eq:x_bound_mismatch} 
show that
\begin{align}
V_q(k_0+k) 
\leq
\frac{
2(\lambda_{\max}(P_q) \bar \alpha_{p,q}^2 +
\rho_q \bar \delta_{p,q}^2 )}{\lambda_{\min}(P_p)}
x(k_0\tau_s)^{\top}P_q x(k_0\tau_s) +
\frac{2(\lambda_{\max}(P_q)\bar \beta_{p,q}^2
+ \rho_q  \bar \gamma_{p,q}^2) }{\rho_p}
\rho_p E_{k_0}^2. \label{eq:Vq_K0+k}
\end{align}
Thus if the switching time $T \in ((k_0+k-1)\tau_s,(k_0+k)\tau_s ]$,
then the bound \eqref{eq:Vq_K0+k} gives
\begin{equation}
\label{eq:Lyapunov_switched}
V_q(k_0+k) \leq \bar \nu_{p,q}(k) V_p(k_0),
\end{equation}
where $\bar \nu_{p,q}(k)$ is defined by
\begin{align*}
\bar \nu_{p,q}(k) &= 
\max
\biggl\{
\frac{
2(\lambda_{\max}(P_q) \bar \alpha_{p,q}(k)^2 +
\rho_q \bar \delta_{p,q}(k)^2) }{\lambda_{\min}(P_p)},~~
\frac{2(\lambda_{\max}(P_q)\bar \beta_{p,q}(k)^2
+ \rho_q  \bar \gamma_{p,q}(k)^2) }{\rho_p}
\biggr\}.
\end{align*}

\subsection{Convergence to the origin}
Finally, we combine 
the average dwell-time property with
the bounds \eqref{eq:Lyapunov_nonswitched} 
and \eqref{eq:Lyapunov_switched} on the Lyapunov functions.

\begin{lemma}
{\em
Define $\nu$ and $\bar \nu$ by
\begin{equation}
\label{eq:nu_bnu_def}
\nu = \max_{p \in \mathcal{P}} \nu_p,\quad
\bar \nu = \max_{p\not=q} \max_{1 \leq k \leq \eta_p}\bar \nu_{p,q}(k).
\end{equation}
If the average dwell time $\tau_a$ satisfies
\begin{equation}
\label{eq:ADT_final_condition}
\tau_a > \left( 1+
\frac{ \log \bar \nu}{\log(1/\nu)} 
\right)\eta\tau_s,
\end{equation}
where $\eta$ is defined by \eqref{eq:eta_def},
then the state converges to the origin, that is, \eqref{eq:convergence_x} holds.
}
\end{lemma}
\begin{proof}
If we have no switches, then convergence to the origin directly follows from
stabilizability of each mode. 
Hence we assume that switches occur.
Fix an integer $M > k_0$. Let the switching times in the interval $(k_0\tau_s, M\tau_s]$
be $T_1,\dots,T_r$. Suppose that $T_i \in ((k_i-1)\tau_s ,k_i\tau_s]$.
By Assumption \ref{ass:switching_time}, $k_{i-1} \leq  k_{i} - 1$
for $i=1,\dots,r$.

Define $\psi_i$ and $\ell_i$ by
\begin{equation*}
\psi_i = \left\lfloor \frac{k_i - k_{i-1} - 1}{\eta_{\sigma(k_{i-1}\tau_s)}} \right\rfloor, \quad
\ell_i = k_{i-1} + \psi_i \eta_{\sigma(k_{i-1}\tau_s)}
\end{equation*}
for $i= 1,\dots,r$. 
Then $\sigma(k_i\tau_s) \not= \sigma(\ell_i \tau_s)$ and
$\sigma(\ell_i \tau_s) = \sigma(k_{i-1}\tau_s)$. 
Moreover, since
\begin{equation}
\label{eq:ceiling_func_property}
\frac{k-n+1}{n} \leq  \left\lfloor \frac{k}{n}\right\rfloor
\leq \frac{k}{n}
\end{equation}
for $k,n \in \mathbb{N}$, it follows that 
$1 \leq k_i - \ell_i \leq \eta_{\sigma(k_{i-1}\tau_s)}$.
This means that we have $\psi_i$ intervals with length 
$\eta_{\sigma(k_{i-1}\tau_s)} \tau_s$ in which no switch occurs
and that the switched case in Section IV. C starts at $t = \ell_i \tau_s$.
We therefore obtain
\begin{equation}
\label{eq:V_switching}
V_{\sigma(k_i\tau_s)}(k_i) 
\leq \bar \nu V_{\sigma(\ell_i\tau_s)}(\ell_i )
\leq \bar \nu \nu^{\psi_i}V_{\sigma(k_{i-1}\tau_s)}(k_{i-1})
\end{equation}
for $i=1,\dots,r$.

Now we investigate the Lyapunov functions after the last switching time $T_r$.
As before, define
\begin{equation*}
\psi_{r+1}= \left\lfloor \frac{M- k_{r} }{\eta_{\sigma(k_{r}\tau_s)}} \right\rfloor, \quad
\ell_{r+1} = k_{r} + \psi_{r+1} \eta_{\sigma(k_{r}\tau_s)}.
\end{equation*}
A discussion similar to the above shows that
\begin{align}
\label{V_final}
V_{\sigma (M\tau_s)}(M) 
\leq \hspace{-1.5pt} 
\hat{\nu} \nu^{\psi_{r+1}}V_{\sigma (k_r \tau_s)}(k_r)
~\text{for some $\hat{\nu}>0$.}
\end{align}


Let us combine
the Lyapunov functions before and after the
last switching time $T_r$.
Define $\psi$ by
\begin{equation}
\label{eq:c_def}
\psi = \sum_{i=1}^{r+1} \psi_i.
\end{equation}
Then \eqref{eq:V_switching} and \eqref{V_final} shows that
\begin{equation}
\label{eq:V_N_bound}
V_{\sigma(M\tau_s)}(M) \leq \hat \nu \bar \nu^{r} \nu^{\psi}V_{\sigma(k_0\tau_s)}(k_0).
\end{equation}
We see from \eqref{eq:ceiling_func_property} that
$\psi$ in \eqref{eq:c_def} satisfies
\begin{align}
\psi
&\geq 
\frac{M - k_0 + 1}{\eta} - (r+1), \label{eq:psi_inequality}
\end{align}
where $\eta$ is defined by \eqref{eq:eta_def}.
Substituting \eqref{eq:psi_inequality} into \eqref{eq:V_N_bound},
we obtain
\begin{equation}
V_{\sigma(M\tau_s)}(M) \leq 
\hat \nu \nu^{1/\eta - 1}  \cdot
\bar \nu^r \nu^{(M-k_0)/\eta - r}
V_{\sigma(k_0\tau_s)}(k_0).
\end{equation}

Suppose that $r = N_{\sigma}(M\tau_s, k_0\tau_s)$ satisfies 
the average dwell-time condition \eqref{eq:ADT_cond}.
Then
\begin{align*}
\bar \nu^r \nu^{(M-k_0)/\eta - r}
\leq
\bar \nu^{N_0}\nu^{-N_0}
\cdot
\left(
\bar \nu^{\tau_s / \tau_a}
\nu^{1/\eta - \tau_s/ \tau_a}
\right)^{M- k_0}.
\end{align*}
Thus if $\tau_a$ satisfies
$\bar \nu^{\tau_s / \tau_a}
\nu^{1/\eta - \tau_s/ \tau_a} < 1$, i.e., \eqref{eq:ADT_final_condition} 
holds,
then we have
$
\lim_{M \to \infty} V_{\sigma(M\tau_s)}(M) = 0.
$

From the convergence of $V_{\sigma}$, 
we easily obtain the desired result \eqref{eq:convergence_x}.
The definition \eqref{eq:Lyapunov_def} of $V_{\sigma(M\tau_s)}$ shows that
$| x(M\tau_s)|$ and $|e(M\tau_s)|_{\infty}$ are bounded by the constant multiplication
of $\sqrt{V_{\sigma(M\tau_s)}(M)}$, and so is $|\xi(M\tau_s)|$.
Hence 
\begin{equation*}
\lim_{M\to \infty}x(M\tau_s) = \lim_{M\to \infty}\xi(M\tau_s) =
0.
\end{equation*}

Since the behavior of $x$ between sampling times is given by
\eqref{eq:x_de_noswitch} and \eqref{eq:x_de_switch},
it follows that
\begin{equation*}
|x(M\tau_s + \tau)| \leq L_1 |x(M\tau_s)|
+ L_2 |\xi(M\tau_s)|
\qquad (0 <  \tau < \tau_s)
\end{equation*}
for some $L_1,L_2 \geq 0$.
Thus the state converges to the origin not only at sampling times
but also in sampling intervals.
\end{proof}

\begin{remark}
{\bf(1)}
To avoid a trivial result, we assume that $\bar \nu \geq 1$. 
Then 
\eqref{eq:ADT_final_condition} implies
\eqref{eq:DTC_zo}, which is the assumption on $\tau_a$ 
at the ``zooming-out'' stage.

\noindent
{\bf(2)}
From \eqref{eq:ADT_final_condition}, we see the relationship between
switching and data rate.
If we increase $N$ in \eqref{eq:N_condition}, 
then $\gamma_p$ defined by \eqref{eq:gamma_p_def}
decreases and hence so do $\nu_p$ in \eqref{eq:nu_p_def}
and $\nu$ in \eqref{eq:nu_bnu_def}.
This leads to a decrease in $\tau_a$. 
%
%
%

\noindent
{\bf (3)}
Piecewise linear Lyapunov functions are also applicable
if an induced norm of $e^{(A_p + B_pK_p)\eta_p \tau_s}$ is less than one
for every $p \in \mathcal{P}$.
For example,
$\|e^{(A_p + B_pK_p)\eta_p \tau_s}\|_{\infty} < 1$ allows us to construct
$V_p = |x|_{\infty} + \rho_pE$ and $V_p = |\xi|_{\infty} + \rho_pE$.
The advantage is that the computation of their upper bounds
are simpler than in
the quadratic case. Such Lyapunov functions may provide less conservative results.
\end{remark}

\section{Lyapnov Stability}
The point here is to find an upper bound on the
finish time of the ``zooming-out'' stage and
an upper bound on the time after which the state with non-zero control input remains in 
$\varepsilon$-neighborhood of the origin at the ``zooming-in'' stage.
Such bounds are dependent on $\tau_a$ and $N_0$ in \eqref{eq:ADT_cond}, 
but not on a switching signal itself.
The former follows from Lemma \ref{lem:ADT_upperbound}  and
the latter proceeds along the
same lines as in Sec. 5.5 of \cite{Liberzon2014}.

Let us first investigate the final time of the ``zooming-out'' stage.

Assume that the average dwell time condition \eqref{eq:ADT_final_condition}
holds, and
let $m$ be an interger satisfying \eqref{eq:N_ADTcond} with 
$\eta \tau_s$ in place of $\tau_0$.
Lemma \ref{lem:ADT_upperbound} with $t_0 = 0$ shows that
for such $m$, there exists 
an integer $n_0 \in[0,(m-1)\eta]$ such that
\[
\sigma(t) = \sigma(n_0\tau_s) =: p
\]
for $t \in [n_0 \tau_s, (n_0+\eta - 1)\tau_s]$.
Moreover, if $\delta$ satisfies
\[
C_{\max} e^{\max_{p \in \mathcal{P}} \|A_p\|_{\infty} 
(m\eta - 1)\tau_s} \delta < \mu_0,
\]
where we define $C_{\max} = \max_{p \in \mathcal{P}} \|C_p\|_{\infty}$,
then $|y(n\tau_s)|_{\infty} \leq \mu_0 \leq \mu_n$, 
and hence $Q_n(y) = 0$ for all $n=0,1,\dots,(m-1)\eta$.
For $E_{n_0}$ defined by \eqref{eq:En0_def},
\eqref{eq:x_j0} shows that
\begin{align*}
|x(n_0\tau_s)|_{\infty} \leq E_{n_0} \leq 
\max_{p \in \mathcal{P}} 
\|W^{\dagger}_p\|_{\infty} \cdot \mu_{m\eta - 1}
=:\bar{E}_0.
\end{align*}
This leads to
\begin{align*}
|x((n_0+\eta_p)\tau_s)|_{\infty} &\leq
e^{\max_{p \in \mathcal{P}} \|A_p\|_{\infty} \tau_s} 
\left\|e^{A_p(\eta_p-1)\tau_s} \right\|_{\infty} 
E_{n_0} (=: E_{n_0+\eta_p}) \\
&\leq 
e^{\max_{p \in \mathcal{P}} \|A_p\|_{\infty} \tau_s} \cdot
\max_{p \in \mathcal{P}} 
\left\|e^{A_p(\eta_p-1)\tau_s}
\right\|_{\infty} \cdot \bar{E}_0
=: \bar E.
\end{align*}
for all $p \in \mathcal{P}$ and $n_0 \in [0,(m-1)\eta]$.
This bound $\bar E$ is independent on a
switching signal itself
and satisfies $E_{n_0+\eta_p} \leq \bar E$.
Also the final time $(n_0+\eta_p)\tau_s$ of the ``zooming-out'' stage
is smaller than $m\eta$, which does depend on $N_0$ and $\tau_a$, but 
not on a switching signal itself.

Next we study the time after which the state with non-zero control input 
remains in $\varepsilon$-neighborhood of the origin at the ``zooming-in'' stage.

The discussion above shows
that the initial time $t_0 := k_0 \tau_s := 
(n_0 + \eta_p)\tau_s$ of the ``zooming-in'' stage
satisfies $k_0 \leq m$ and that
$|x(k_0\tau_s)|_{\infty} \leq E_{k_0} \leq \bar E$.
Define $\Lambda_{\max}$ by
\[
\Lambda_{\max} = \max_{p\in \mathcal{P}} (
{\sf n} \lambda_{\max}(P_p) + \rho_p).
\]
Since $e(k_0\tau_s) = x(k_0\tau_s)$, it follows that
the Lyapunov function in \eqref{eq:Lyapunov_def} satisfies
$
V_{\sigma(k_0)\tau_s} \leq \Lambda_{\max} \bar E^2.
$
Thus we see from \eqref{eq:V_N_bound} that, for all $M \geq k_0$
\begin{equation}
\label{eq:V_bound_forLyapnouvS}
V_{\sigma(M\tau_s )}(M) \leq \hat \nu \bar \nu^{N_{\sigma}(M\tau_s,k_0\tau_s)} 
\nu^{\psi}\Lambda_{\max} \bar E^2.
\end{equation}
In conjunction with \eqref{eq:ADT_final_condition},
this shows that, for every $\varepsilon > 0$, there exists $M_0 \geq 0$ such that
$V_{\sigma(M\tau_s)} (M)< \varepsilon$ for $M \geq M_0$.
Hence for every $\varepsilon > 0$, there also exists $T_0 \geq 0$ such that
\begin{equation}
\label{eq:x_bound_Lyapunov_afterT_0}
|x(t)|_{\infty} < \varepsilon \qquad  (t \geq T_0).
\end{equation}
Notice that \eqref{eq:V_bound_forLyapnouvS} implies that 
$M_0$ depends on $N_0$, $\tau_a$, and $\varepsilon$ 
but not on a switching signal itself, and so does $T_0$.

If we obtain
\begin{equation}
\label{eq:x_bound_Lyapunov_beforeT_0}
|x(t)|_{\infty} < \varepsilon \qquad  (t \leq T_0),
\end{equation}
then
combining \eqref{eq:x_bound_Lyapunov_afterT_0} and
\eqref{eq:x_bound_Lyapunov_beforeT_0}
completes the proof for Lyapunov stability.

We see from
\eqref{eq:En0+etap},
\eqref{eq:quantization}, and \eqref{eq:e_noswitch_bound} that
if $\delta$ satisfies
\begin{align*}
C_{\max} e^{\max_{p \in \mathcal{P}} \|A_p\|_{\infty} T_0} \delta &\leq 
\frac{1}{\sf p}
\min_{p \in \mathcal{P}} \min_{0\leq k \leq \eta_p-1} 
\| C_p e^{A_pk\tau_s}\|_{\infty} \cdot
(\min_{p \in \mathcal{P}} \theta_p)^{\lfloor T_0 / \min_{p \in \mathcal{P}}
\eta_p \rfloor} \\
&\qquad
\times \min_{p \in \mathcal{P}}\left( 
\left\|e^{A_p(\eta_p - 1)\tau_s}\right\|_{\infty} \cdot \|W^{\dagger}_p\|_{\infty}
\right)\mu_0,
\end{align*}
then the quantized output $q_k$ is zero for $kT_s \leq T_0$. 
This means that $\xi(t) = u(t) = 0$ for $t \leq T_0$.
Therefore if $\delta$ additinally satisfies 
\[
e^{\max_{p \in \mathcal{P}} \|A_p\|_{\infty} T_0} \delta 
<\varepsilon,
\]
then we have \eqref{eq:x_bound_Lyapunov_beforeT_0}.

In Fig.~\ref{fig:lyap},
we illustrate the state trajectory for the Lyapunov stability.

 \begin{figure}[t]
 \centering
 \includegraphics[width = 7cm,clip]{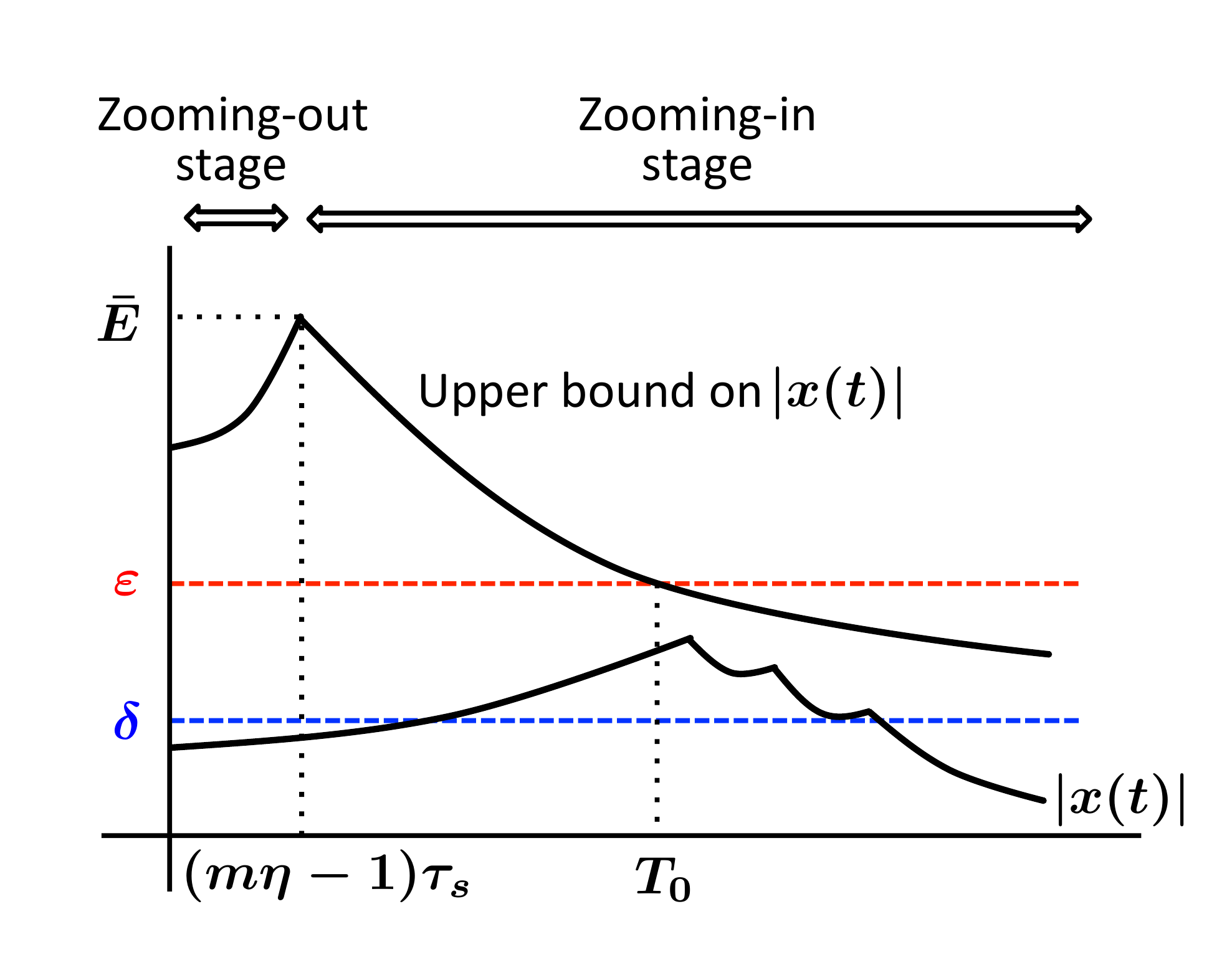}
 \caption{The state trajectory for the Lyapunov stability}
 \label{fig:lyap}
 \end{figure}%
%
%
%
%

\section{Numerical Example}
Consider a continuous-time switched system \eqref{eq:SLS} 
with the following two modes:
\begin{align*}
A_1 = 
\begin{bmatrix}
0 & -1 \\ -1 & -2
\end{bmatrix},\quad
&B_1 = 
\begin{bmatrix}
1 \\ -1
\end{bmatrix},\quad
C_1 = 
\begin{bmatrix}
1 & 1
\end{bmatrix}\\
A_2 = 
\begin{bmatrix}
1 & 2 \\ -2 & -1
\end{bmatrix},\quad
&B_2 = 
\begin{bmatrix}
-2 \\ 1
\end{bmatrix},\quad
C_2 = 
\begin{bmatrix}
1 & -1
\end{bmatrix}.
\end{align*}
The feedback gains of each mode are
$K_1 = [-1~~2]$ and $K_2 = [1~~-1]$.
Note that $A_1+B_1K_2$ and $A_2 + B_2K_1$
have an unstable pole $2.2361$ and $4$,
respectively.
The sampling period $\tau_s$ and
the partition number $N$ of the quantizer are
$\tau_s = 0.5$ and $N = 11$.

We took $Q_1$ and $Q_2$ in \eqref{eq:Lyapnov_equation} and
the parameters of $\nu_1$ and $\nu_2$ in \eqref{eq:nu_p_def}
as follows: 
$Q_1 = Q_2 = I$, 
$\kappa_1 = 1.124$, $\kappa_2 = 1.09$, 
$\rho_1 = 47$, and $\rho_2 = 80$. These were chosen by
trial and error.
We see from \eqref{eq:ADT_final_condition} that
if $\tau_a > 5.55$, our encoding and control strategy
achieves the global asymptotic stabilization.

A time response in the interval $[0,20]$ was calculated for
$x(0) = [-3~~3]^{\top}$, $\mu_0 = 0.1$, and $\chi = 1$.
The switching signal was chosen so that 
the dwell time $\tau_d = 2.6$ and
\eqref{eq:ADT_cond}
holds with $N_0 =1$ and the average dwell time $\tau_a = 5.8$.
Fig.~\ref{fig:simulation}
shows the Euclidean norm of
the state $x$ and the state estimate $\xi$. 
In this simulation, the ``zooming-out'' stage finished at $t=1$ but we observe that 
the system was not controlled until $t=2$. The reason is that
the state estimate is zero at the end of ``zooming-out'' stage; see \eqref{eq:ES_Initial}.
This leads to no control input in
the initial period of ``zooming-in'' stage as in
the ``zooming-out'' stage.

If the state of the plant is accessible, i.e., $C_1 = C_2 = I$, then we see from 
\cite[Assumption 3]{Liberzon2014} that if the number of symbols in the 
quantizer is not smaller than $3^2+1 = 10$, then 
the encoding and control strategy in \cite{Liberzon2014}
stabilizes the plant.
On the other hand, the counterpart in the output feedback case from
\eqref{eq:N_condition} is $5+1 = 6$.
Hence in this example, if we consider only stabilization
of systems with sufficiently large average-dwell time property, then 
the use of the output with lower dimension than the dimension of the state
has the advantage in terms of data rate.
 \begin{figure}[t]
 \centering
 \includegraphics[width = 7cm,bb= 35 20 700 510,clip]{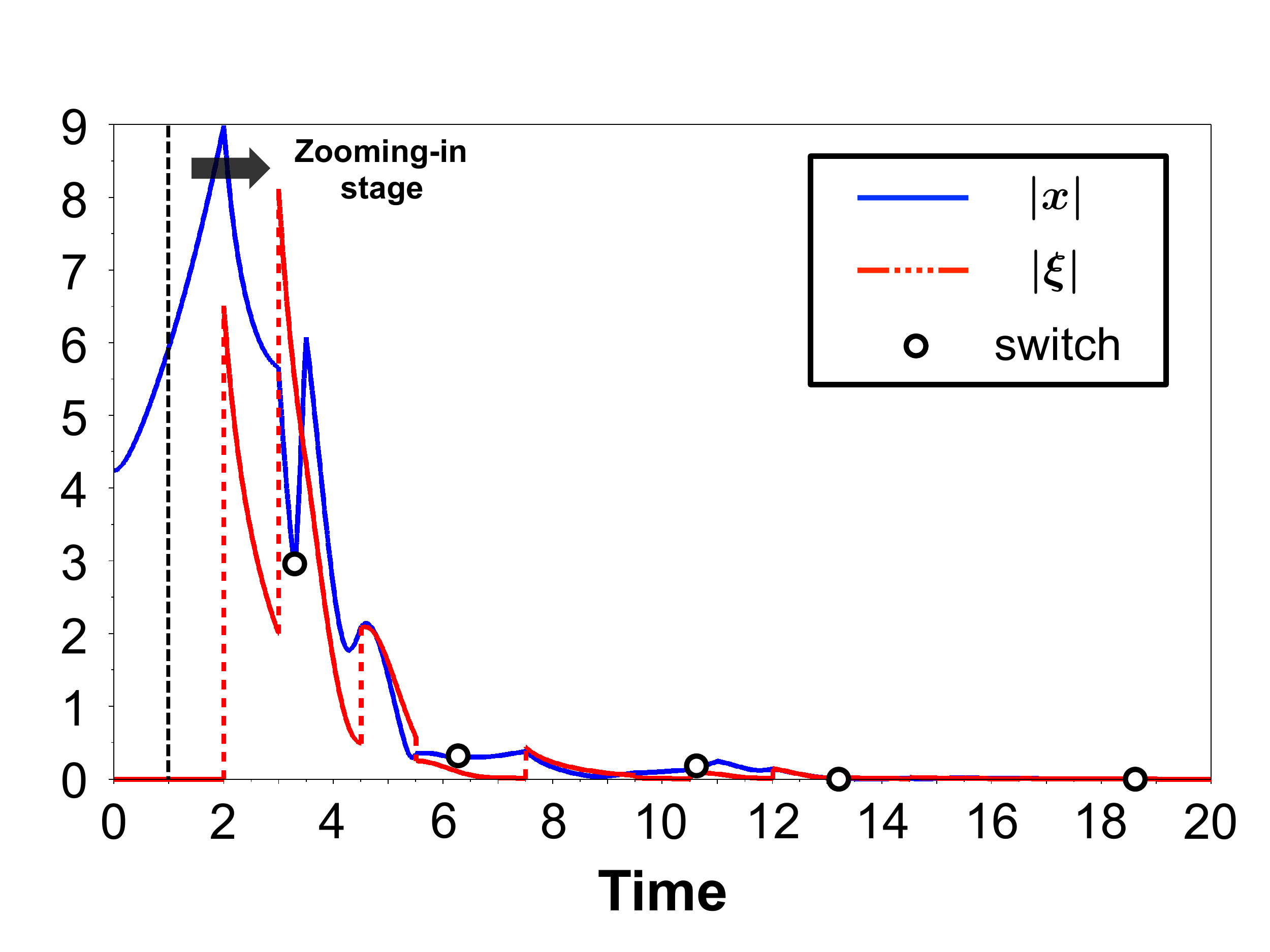}
 \caption{The Euclidean norm of the state $x$ and the estimated state $\xi$}
 \label{fig:simulation}
 \end{figure}%
%


\section{Concluding Remarks}
We have studied the problem of stabilizing 
a switched linear system with limited information: 
the quantized output and active mode at each sampling time.
We have supposed that the controller is given
and have examined the intersample behavior of the estimation error
for the encoding strategy after the detection of switching.
Using multiple discrete-time Lyapunov functions, we have achieved
global asymptotic stabilization under
the hybrid dwell-time assumption.
The data-rate bound used here is
the maximum among the bounds of the individual subsystems
that are from the earlier work.


\end{document}